\documentclass[lettersize,journal]{IEEEtran}
\usepackage{amsmath,amsfonts}
\usepackage{algorithmic}
\usepackage{algorithm}
\usepackage{array}
\usepackage[caption=false,font=normalsize,labelfont=sf,textfont=sf]{subfig}
\usepackage{todonotes}
\usepackage{textcomp}
\usepackage{stfloats}
\usepackage{url}
\usepackage{verbatim}
\usepackage{graphicx}
\usepackage{cite}
\usepackage{amsthm}
\theoremstyle{definition}

\newtheorem{theorem}{Theorem}

\renewcommand{\P}{\mathbb{P}}

\newcommand{\E}{\mathbb{E}}

\def\mcD{{\mathcal D}}

\def\mcO{{\mathcal O}}

\usepackage{tabularx,graphicx}

\newcommand{\var}{\hbox{VAR}}

\newcommand{\ar}[1]{\textcolor{black}{#1}}

\begin{document}

\title{Fundamental Scaling Laws of Covert Communication in the Presence of Block Fading}

\author{Amir Reza Ramtin,~\IEEEmembership{Member,~IEEE,}\thanks{This work was supported in part by the National Science Foundation under Grant ECCS-2148159.} Dennis Goeckel,~\IEEEmembership{Fellow,~IEEE,} and Don Towsley,~\IEEEmembership{Fellow,~IEEE} 
\thanks{A. Ramtin and D. Towsley are with the Manning College of Information and Computer Sciences, University of Massachusetts at Amherst, MA, USA (email: aramtin@cs.umass.edu, towsley@cs.umass.edu).}
\thanks{D. Goeckel is with the Department of Electrical and Computer Engineering, University of Massachusetts at Amherst, MA, USA (email: dgoeckel@engin.umass.edu).}
}

\maketitle

\begin{abstract}
	
	Covert communication
 is the undetected transmission of sensitive information over a communication channel. 
 In wireless communication systems, channel impairments such as signal fading 
 present challenges in the effective implementation and analysis of covert communication systems. This paper generalizes early work in the covert communication field by considering asymptotic results for the number of bits that can be covertly transmitted in $n$ channel uses on a block fading channel.  Critical to the investigation is characterizing the performance of optimal detectors at the adversary.  Matching achievable and converse results are presented.
	
\end{abstract}


\section{Introduction}

Covert communication is the transmission of information between two parties, traditionally known as Alice and Bob, such that the presence of the transmission is undetectable by a third-party observer (adversary) named Willie. It is particularly applicable in military and security domains, where securely transferring sensitive information without detection is of utmost importance \cite{bloch2008wireless, liang2017information}.  A crucial metric in the domain of covert communication is
the number of bits that can be reliably and covertly transmitted in $n$ uses of the channel \cite{bash2013limits,sobers2017covert}.  Evaluating the performance of covert communication systems traditionally involves assessing decoding error probabilities at receiver Bob and detection error probabilities at Willie. These probabilities depend on the specific communication scenario, such as the presence of noise, fast fading, and the detection strategy Willie employs.  It can be particularly difficult to characterize Willie's optimal strategy so as to characterize achievable results for covert communication between Alice and Bob \cite{sobers2017covert}.

In \cite{tahmasbi2019covert}, the authors characterized the maximum number of bits that can be covertly transferred in a non-coherent fast Rayleigh-fading wireless channel, showing that this maximum is achieved with an amplitude-constrained input distribution having a finite number of mass points, including zero. It provides numerically tractable bounds and conjectures that two-point distributions may be optimal. However, the maximum number of bits that can be covertly transferred in this channel still adheres to the square root law (SRL), as in \cite{bash2013limits}.  When a friendly jammer is present, the number of bits that can be transmitted covertly increases, and \cite{sobers2017covert} characterized scaling laws for achievable results for both additive white Gaussian noise (AWGN) and block fading channels with a constant number of fading blocks (i.e., a number of fading blocks not scaling with $n$) in the presence of such a jammer.  \cite{goeckel2018covert} generalizes  \cite{sobers2017covert} to the case of covert communication in the presence of a jammer with an arbitrary number of fading blocks per codeword, 
and provides upper and lower bounds to the covert throughput.

Here we consider scaling results for covert system performance 
under a block fading model 
in the absence of jamming, \ar{assuming that the number of fading blocks grows with $n$, without 
prescribing a specific growth rate.}  Although \ar{our} model differs from that considered in \cite{sobers2017covert}, a similar challenge arises:  it is not easy to characterize the performance of the optimal detector at Willie.  When faced with such a challenge, authors often assume the adversary employs a reasonable but potentially sub-optimal detector such as a power detector, and then characterize covert system performance \cite{chen2023covert}.  
However, as illustrated in Section~\ref{sec:eval}, the power detector differs from the optimal receiver in the presence of block fading, which underscores the need to specifically consider the optimal detector to investigate scaling laws of covert communication.

The work of \cite{sobers2017covert} indicates that random jamming, the impairment considered there, has an asymmetric effect: it inhibits Willie's ability to detect the signal more than Bob's ability to decode it, hence improving covert throughput.  Hence, in the case considered here,
we might expect a similar effect:  that block fading inhibits Willie's ability to detect a transmission from Alice, allowing Alice to exceed the AWGN case of \cite{bash2013limits} by using power $\mcO(n^{-\alpha})$ and achieving throughput $\mcO(n^{1-\alpha})$ for some $\alpha<0.5$. 
This possibility is rigorously investigated through the converse analysis presented in Theorem~\ref{th:converse}. Investigating achievability, as we do in Theorem~\ref{th:ach}, is also necessary and non-trivial.
Perhaps surprisingly, we show that $\mcO(\sqrt{n})$ remains the fundamental limit.

 While various works in covert communication have explored block fading environments,
\ar{they typically focus on specific achievable schemes or detector performance under particular fading assumptions rather than fundamental scaling laws under optimal detection. More specifically, \cite{sobers2017covert} studies covert communication in the presence of an uninformed jammer, while \cite{goeckel2018covert} considers a dynamic interference environment. In contrast, the present work studies block fading in the absence of jamming and asks whether block fading alone changes the fundamental square root scaling law. Existing works on covert communication over fading channels, such as \cite{shahzad_covert_2017,hu2018covert,ta_covert_2019}, consider specific fading and detector settings, whereas our focus is on matching achievability and converse scaling laws under optimal adversarial detection. As discussed above,} 
it remains unclear whether the square root law observed in AWGN channels applies to this setting, given the distinct challenges in detection accuracy under block fading. Our work addresses this open question, providing a foundation for understanding covertness in more realistic fading environments. 

While the persistence of the square-root law may be expected in hindsight, our contribution is to rigorously establish its robustness under block fading with 
\ar{a growing} number of fading blocks, optimal adversarial detection, and regardless of the adversary's knowledge of the distribution at $\mathrm{H_1} $.

The paper is organized as follows: Section~II defines the problem and system model, Section~III presents the main theoretical results, Section~IV provides numerical evaluations, and Section~V concludes the paper.
\section{Problem Definition and Framework}

\subsection{System Overview}

Figure~\ref{fig:abw}(a) provides a schematic of the covert communication system. It includes Alice 
($a$), who may attempt to send a message to her intended recipient, Bob 
($b$), in such a way that it remains undetected by the observant and capable adversary, Willie 
($w$), as defined below.
\begin{figure}[h]
	\centering
	\begin{tabular}{cc}
		\includegraphics[width=0.10\textwidth]{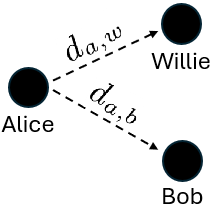} &
		\includegraphics[width=0.25\textwidth]{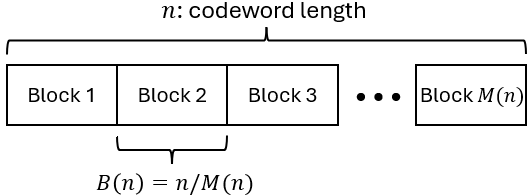}\\
		(a) & (b)
	\end{tabular}
	\caption{(a) Alice ($a$) is trying to send messages to Bob ($b$) without detection by a capable and observant Willie ($w$). Here, $d_{x, y}$ stands for the distance from one transmitter $x$ to a receiver $y$. Over the length of a codeword, denoted by $n$, signal fading occurs independently $M(n)$ times, resulting in a block of length $B(n)=n / M(n)$, within which the fading remains constant. (b) Over the length of a codeword, denoted by $n$, signal fading occurs independently $M(n)$ times, resulting in a block of length $B(n)=n / M(n)$, within which the fading remains constant.}
	\label{fig:abw}
\end{figure}

The framework for our study is grounded on the block fading model \cite{tse2005fundamentals}. In this model, the length of a codeword, denoted as $n$, is divided into $M(n)$ separate fading blocks, as shown in Figure~\ref{fig:abw}(b). 
Within each block, fading is constant over its length, while fading across different blocks is independent. The block length $B(n)=n/M(n)$ can be interpreted as the channel coherence interval, i.e., the number of channel uses over which the fading coefficient remains constant.
\ar{Throughout the paper, we assume $M(n)\to\infty$ as $n\to\infty$, while no specific growth rate for $M(n)$ is imposed.}

In each fading block $ k $, the complex fading between a transmitter $ x $ and a receiver $ y $ is represented as $ h_k^{(x, y)} $. In our context, $ x $ is Alice ("$a$"), while $ y $ is either Bob ("$b$") or Willie ("$w$"). We assume Rayleigh fading, which implies that the fading coefficient $ h_k^{(x, y)} $ is a zero-mean complex Gaussian random variable with $ \E[|h_{x, y}|^2] = 1 $.

Alice decides whether to transmit ($\mathrm{H_1} $) or refrain from transmitting ($ \mathrm{H_0} $), with an equal chance of each of the two options. 
 When Alice decides to transmit, her transmitted signal is given by 
 symbols $ f_1, f_2, ..., f_n $.

For Willie, the received signal is ($i=1,2,\ldots ,n$):
\[ 
Z_i = \begin{cases} 
\frac{h_{\lceil i / B(n)\rceil}^{(a, w)} f_i}{d_{a, w}^{\alpha/2}} + N_i^{(w)}, & \text{Alice transmits,} \\ 
N_i^{(w)}, & \text{otherwise,}
\end{cases} 
\]
where $ \alpha $ represents the pathloss exponent, and  $ \{N_i^{(w)}\}_{i=1}^n
$ is a sequence of i.i.d. complex Gaussian random variables, each with zero mean and variance $ \sigma_0^2 $.  

Similarly, for Bob, the received signal is modeled as:
\[ 
V_i = \begin{cases} 
\frac{h_{\lceil i / B(n)\rceil}^{(a, b)} f_i}{d_{a, b}^{\alpha/2}} + N_i^{(b)}, & \text{Alice transmits,} \\ 
N_i^{(b)}, & \text{otherwise,}
\end{cases} 
\] 
where $ \{N_i^{(b)}\}_{i=1}^n
$ 
is an i.i.d. sequence of complex Gaussian random variables, each with mean zero and variance $ \sigma_b^2 $.


\subsection{Covert Rate Asymptotics and Objectives}

The goal of this study is to derive asymptotic bounds on the number of bits that can be covertly transmitted in a given number of channel uses, in the presence of block fading between Alice and Willie, and Alice and Bob. 
Here
Willie employs statistical hypothesis testing to distinguish between the null hypothesis $\mathrm{H_0}$, that Alice has not transmitted, and the alternative hypothesis $\mathrm{H_1}$ where she has.
Willie's goal is to minimize the probability of missed detection, $\P_{MD}$, while ensuring a minimum false alarm rate, $\P_{FA}$. 
The system is covert \cite{bash2013limits} if $\P_{FA} + \P_{MD} \ge 1-\epsilon$ for any $\epsilon > 0$ when $n$ is sufficiently large.

\section{Main Results}
\label{sec:model}

We first present preliminaries on covertness and hypothesis testing, followed by achievability and converse results.

\subsection{Preliminaries}
Consider an i.i.d. sequence of 
complex variables $\{Z_i\}_{i=1}^n$ representing the signals received by Willie. The distribution of the observations under each hypothesis, $\mathrm{H_0}$ and $\mathrm{H_1}$, is given by $\P_0$ and $\P_1$, respectively.

Given that distribution $\P_0$ is known, Willie can formulate an optimal statistical hypothesis test (e.g., the Neyman-Pearson test given $\P_1$ is also known). This test minimizes the sum of error probabilities \cite[Ch. 13]{LR05}
\begin{equation} \P_E = \P_{FA} + \P_{MD}.\end{equation} 
In our analysis, we consider both known and unknown $\P_1$.

For such a test scenario, the following theorem holds:
\ar{\begin{theorem}\cite{bash2013limits}
\label{thm:optimal-test}
	\begin{equation}
	 \P_E \ge 1-\sqrt{\tfrac{1}{2}\mcD (\P_1||\P_0)},
	\end{equation}
\end{theorem}
\noindent where $\mcD (\P_1||\P_0)$ denotes Kullback-Leibler (KL) divergence between two probability measures $\P_1$ and $\P_0$.}

\subsection{Lower Bound}

\begin{theorem}[Lower bound]
	\label{th:ach}
	Under the fading model introduced previously, given a total of $n$ channel uses, Alice can transmit $\Omega(\sqrt{n})$\footnote{In covert communications, the lower bound is often expressed as an achievability result using $\mcO(\cdot)$. However, we find that $\Omega(\cdot)$ is more appropriate, as it explicitly provides a lower bound on achievable performance.} bits
    while remaining covert. This remains valid regardless of whether or not Willie knows $\P_1$.
\end{theorem}
\begin{proof} 
	
Let $\mathrm {Re}(z)$ and $\mathrm{Im}(z)$ denote the real and imaginary parts of $z \in \mathbb{C}$, respectively. Define 
\begin{equation}
\label{def:fztheta}
f(z,\theta) = f_g(\mathrm {Re}(z),\theta/2)f_g(\mathrm{Im}(z),\theta/2),
\end{equation}
where $f_g(x,\theta)$ is a Gaussian density function with mean 0 and variance \( \theta \), i.e.,
$f_g(x,\theta)=\frac{1}{\sqrt{2\pi \theta}}e^{-\frac{x^2}{2\theta}}$. 

Define the realization of the random variables $\{Z_{(i-1)B(n)+j}\}_{j=1}^{B(n)}$ with $i=1,2,\ldots,M(n)$ as vector $\mathbf{z}_i = (z_{(i-1)B(n)+1}, z_{(i-1)B(n)+2}, \ldots, z_{iB(n)})$, where $B(n)$ represents the block length within each fading block and $M(n)$ denotes the number of blocks.
As specified in the main channel throughput argument later, Alice draws her message from an i.i.d. Gaussian codebook; hence, $f_1,
f_2,...,f_n$ are i.i.d. random variables, where $f_i\sim\mathcal{CN}(0,\sigma_a^2(n))$ for $i=1,2,\dots,n$. Let $\lambda = d_{aw}^\alpha$.
Then, the probability density function of the $i$th block under $\mathrm{H_0}$ is given by
\begin{equation}
	\label{eq:fading:f0}
f_0(\mathbf{z}_i)  = \textstyle \prod_{j=1}^{B(n)} f(\mathbf{z}_{i,j},\sigma_0^2),
\end{equation}
and under $\mathrm{H_1}$ is given by
\begin{equation}
	\label{eq:fading:f1}
f_1 (\mathbf{z}_i)  = \textstyle \int_0^\infty \prod_{j=1}^{B(n)} f(\mathbf{z}_{i,j},\sigma_0^2+x\sigma_a^2(n)) \lambda e^{-\lambda x}dx.
\end{equation}

The KL divergence 
is not easily computed \ar{in closed form}.
However, 
\ar{since $f_1(\mathbf z_i)$ is a mixture over the nonnegative variable $x$, while $f_0(\mathbf z_i)$ is fixed, the convexity of relative entropy \cite{CT02} yields
\begin{equation}
\begin{aligned}
&\mcD\big(f_1(\mathbf{z}_i)||f_0(\mathbf{z}_i)\big)\le \\ &  \textstyle \int_{0}^{\infty} \mcD\big(f(\mathbf{z}_{i,j},\sigma_0^2+\sigma_a^2(n)x)^{B(n)}\|f(\mathbf{z}_{i,j},\sigma_0^2)^{B(n)}\big) \lambda e^{-\lambda x} d x\\
& = \textstyle B(n)\int_{0}^{\infty}\mcD\big(f(\mathbf{z}_{i,j},\sigma_0^2+\sigma_a^2(n)x)\|f(\mathbf{z}_{i,j},\sigma_0^2)\big) \lambda e^{-\lambda x} d x,
\end{aligned}
\end{equation}}
where the last equality follows from the additive property of relative entropy for independent distributions.

Following \eqref{def:fztheta} and employing the same additive property,
\begin{equation}\begin{aligned}
&\mcD\big(f(\mathbf{z}_{i,j},\sigma_0^2+\sigma_a^2(n)x)\|f(\mathbf{z}_{i,j},\sigma_0^2)\big)\\
& \quad=\mcD\big(f_g(\mathrm {Re}(\mathbf{z}_{i,j}),\sigma_0^2/2+\sigma_a^2(n)x/2)\|f_g(\mathrm {Re}(\mathbf{z}_{i,j}),\sigma_0^2/2)\big)\\
&\quad+ \mcD\big(f_g(\mathrm{Im}(\mathbf{z}_{i,j}),\sigma_0^2/2+\sigma_a^2(n)x/2)\|f_g(\mathrm{Im}(\mathbf{z}_{i,j}),\sigma_0^2/2)\big) \\
&\quad=-\log\! \big(1+ x{\sigma_a^2(n)}/{\sigma_0^2}\big) +x{\sigma_a^2(n)}/{\sigma_0^2}.
\end{aligned}\end{equation}
This implies that
\begin{equation}\begin{aligned}
&\mcD\big(f_1(\mathbf{z}_i)||f_0(\mathbf{z}_i)\big)\\
& \le 
\textstyle {B(n)}\int_0^\infty \big(-\log\!\big(1+ \tfrac{\sigma_a^2(n)}{\sigma_0^2}x\big) +\tfrac{\sigma_a^2(n)}{\sigma_0^2}x\big)\lambda e^{-\lambda x}\,dx 
\end{aligned}\end{equation}

Since $\int_0^{\infty}\log(1+cx)\lambda e^{-\lambda x}\,dx=-e^{\lambda/c}\operatorname{Ei}(-\lambda/c)$,
where 
$\operatorname{Ei}(x)= -\int_{-x}^\infty \frac{e^{-t}}t\,dt$,
\begin{equation}\begin{aligned}  
\mcD\big(f_1(\mathbf{z}_i)||f_0(\mathbf{z}_i)\big) & \le {B(n)}e^{\frac{\lambda\sigma_0^2}{\sigma_a^2(n)}}\operatorname{Ei}\!\big(-\tfrac{\lambda\sigma_0^2}{\sigma_a^2(n)}\big) + \tfrac{B(n)\sigma_a^2(n)}{\lambda\sigma_0^2}.
\end{aligned}\end{equation}

Note that $\operatorname{Ei}(-x)< -\frac12e^{-x}\log\!\big(1+\frac2x\big)$ and $\log(1+x)\ge x-\frac{x^2}2$, which together imply $\operatorname{Ei}(-x)< -\frac12e^{-x}\big(\frac2x-\frac2{x^2}\big)$. Using this bound and simplifying yields 
\begin{equation}\begin{aligned}
\mcD\big(f_1(\mathbf{z}_i)||f_0(\mathbf{z}_i)\big)
& < 
\tfrac{B(n)\sigma_a^4(n)}{\lambda^2\sigma_0^4}.
\end{aligned}
\label{eq:DKL}
\end{equation}

From Theorem~\ref{thm:optimal-test}, it suffices to find the fastest growing $\sigma_a(n)$ such that
$\lim_{n\rightarrow \infty}M(n)\mcD\big(f_1({\mathbf{z}_i})||f_0({\mathbf{z}_i})\big) <\infty,$
which establishes 
$M(n)\frac{B(n)\sigma_a^4(n)}{\lambda^2\sigma_0^4} = \mcO(1).$
This implies that $\sigma_a^2(n)$ can be chosen as $\mcO\big(1/\sqrt{M(n)B(n)}\big)=\mcO(1/\sqrt{n}).$
Thus Alice can transmit with power $\mcO\big(1/\sqrt{n)}$ 
with this scheme. As there may be better performing schemes, this establishes that Alice can transmit with power $\Omega(1/\sqrt{n})$.

This result is built on 
an optimal detector
with knowledge of both $\P_0$ and $\P_1$. It remains valid even if Willie lacks knowledge of $\P_1$, as Alice's performance can only improve.

Since the Alice-to-Bob communication channel experiences block fading---unlike the AWGN channels typically assumed in reliable communication scenarios---there remains a nonzero probability, not diminishing with $n$, that some sub-blocks will experience deep fades, potentially causing decoding failures at Bob \cite{sobers2017covert}. To rigorously ensure reliable decoding despite block fading, we adapt Boulat Bash et al.'s coding argument \cite{bash2013limits} as follows.

We define a threshold $\delta > 0$ and treat blocks with fading power $|h^{(a,b)}|^2 < \delta$ as effectively unusable (deep fades). The expected fraction of usable (good) blocks is then \ar{$\gamma := \int_{\delta}^{\infty} e^{-x}\,dx = e^{-\delta}$}.

To compensate for blocks that are
unusable, we instead let Alice transmit over $n$ symbols total, but design the code such that only a subset $n' = \Theta(n)$ of the $n$ symbols need to land in 
good blocks. Each symbol is chosen independently from $\mathcal{N}(0,\sigma_a^2)$ with power scaling $\sigma_a^2 = c/\sqrt{n}$ with $c>0$. Since this power scaling satisfies $\sigma_a^2\in\mcO(1/\sqrt{n})$, the previously established argument directly implies covertness under block fading.

Partition the $n$ symbols into $M(n)=n/B(n)$ blocks, each experiencing independent fading. Let $G_j=\mathbf{1}_{{|h_j^{(a,b)}|^2\ge\delta}}$ indicate good
blocks, and define the number of good blocks as $S=\sum_{j=1}^{M(n)}G_j$. Recall that we allow $M(n)$ to grow without bound as $n \to \infty$, ensuring that Hoeffding’s inequality can be applied to the i.i.d. $\mathrm{Bernoulli}(\gamma)$ indicators $G_j$, yielding
$\P\big(|S-\gamma M(n)|>\epsilon_2\gamma M(n)\big)\le 2e^{-2\epsilon_2^{2}\gamma M(n)},$ for any $\epsilon_2>0$. Thus, with probability approaching one
as $M(n)\to \infty$,
Bob observes at least $n' := (1-\epsilon_2)\gamma n$ good (high-SNR)
symbols as $n\to\infty$.

Because only the $n'$ good symbols contribute, the squared Euclidean distance between any two distinct codewords---restricted to these coordinates---is approximately $2\sigma_a^{2}n' = 2c' \sqrt{n'}$, where $c' = c \sqrt{(1 - \epsilon_2)\gamma}$. This leads to a pairwise error probability upper bound $\exp\!\big(-\frac{c'}{4\sigma_b^2}\sqrt{n'}\big)$. Applying the union bound over all $2^{n'R}$ messages yields an overall error probability of $\exp\!\big(n'R\ln 2 - \frac{c'}{4\sigma_b^2}\sqrt{n'}\big)$. To ensure this exponent is negative, we choose $R = \frac{\rho}{\sqrt{n'}}$ with $\rho \in \big(0, \frac{c'}{4\sigma_b^{2} \ln 2}\big)$. Each codeword encodes a message from among $2^{n'R}$ possibilities, thus reliably communicating $\log(2^{n'R}) = n'R = \Theta(\sqrt{n'}) = \Theta(\sqrt{n})$ covert information bits.


\ar{Since our scheme achieves this rate with power scaling $\Theta(1/\sqrt{n})$, and, as previously mentioned, more efficient schemes may exist, Alice can reliably and covertly transmit $\Omega(\sqrt{n})$ bits in $n$ channel uses.}

\end{proof}

\subsection{Converse (Upper Bound)}
\begin{theorem}[Upper Bound]
    \label{th:converse}
	Assume the fading model defined earlier with a total of $n$ symbols. 
    Then, Alice is restricted to $\mcO(\sqrt{n})$ bits in $n$ channel uses.
    This holds whether or not Willie knows $\P_1$.
\end{theorem}

\begin{proof}
	Consider an arbitrary codebook. Denote the total transmitted power across
	all symbols for the codeword chosen for transmission by $P_n$, i.e., $\sum_{i=1}^n |f_i|^2 = P_n$.   Denote the ratio of power allocated to block $j$ to total power by $w_j$, i.e., $w_j = \sum_{k=1}^{B(n)} |f_{(j-1)B(n)+k}|^2/P_n$. Note that $\sum_{j=1}^{M(n)} w_j = 1$. We will show that $P_n$ must be limited, or the presence of the sequence on the channel is readily detected.
	
	Define the test statistic $\bar{Y} = \frac{1}{n}\sum_{i=1}^{n} Y_i$ where $Y_i=|Z_{i}|^2$ with $i=1,\ldots,n$.

	Consider a threshold-based test to distinguish between $\mathrm{H_0}$ and $\mathrm{H_1}$. Let the threshold be $t = \E_0[\bar{Y}] + d/\sqrt{n}$, where $d > 0$. 
	
    Under $\mathrm{H_0}$, 
	\begin{equation}
	\P_{FA}=\P_0(\bar{Y}>t ) = \P_0(\bar{Y}-\E_0[Y_i] > d/\sqrt{n}).
	\end{equation}
    Also, $Z_i \sim \mathcal{CN}(0, \sigma_0^2)$, and thus $\E_0[\bar{Y}]=\sigma_0^2$ and $\var_0(\bar{Y})=\sigma_0^4/n.$
	Applying Chebyshev’s inequality,
	\begin{equation}
	\P_{FA}\le{\var_0(\bar{Y})}/{(d/\sqrt{n})^2}={\sigma_0^4}/{d^2}.
	\end{equation}
	Thus, for any desired $\P_{FA}>0$, we can set $d=\sigma_0^2/\sqrt{\P_{FA}}$ to ensure the bound holds.

Under $\mathrm{H_1}$, 
by
Chernoff's inequality, for any $s\in(0,1/\sigma_0^{2})$,
\begin{equation}
\begin{aligned}
\P_{MD}
= \P_1(\bar Y \le t)
\le e^{snt}\,\E_1\big[e^{-s\sum_{i=1}^n Y_i}\big].
\end{aligned}
\end{equation}
Denote $\lambda = d_{aw}^{\alpha}$ and $X_j=|h_j|^2/\lambda$. Then, since the blocks are independent, 
\begin{equation}
\E_1\big[e^{-s\sum_{i=1}^n Y_i}\big]
= \textstyle \prod_{j=1}^{M(n)} \E_{X_j}\Big[\E\big[e^{-s\sum_{k=1}^{B(n)} Y_{(j-1)B(n)+k}} \,\big|\,X_j\big]\Big].
\end{equation}
For $i=(j-1)B(n)+k$ in block $j$, under $\mathrm{H_1}$ we have
$Z_i \mid X_j \sim \mathcal{CN}(\sqrt{X_j}\,f_i,\sigma_0^2)$, so
$Y_i=|Z_i|^2$ has conditional Laplace transform
$\E[e^{-s Y_i}|X_j]={e^{-\frac{s\,X_j |f_i|^2}{1+\sigma_0^2 s}}}/{(1+\sigma_0^2 s)}$.
Hence, within block $j$, using $\sum_{k=1}^{B(n)} |f_{(j-1)B(n)+k}|^2 = w_j P_n$,
\begin{equation}
\E\big[e^{-s\sum_{k=1}^{B(n)} Y_{(j-1)B(n)+k}}\big|X_j\big]
=
\frac{e^{-\frac{s X_j}{1+\sigma_0^2 s}w_j P_n}}{(1+\sigma_0^2 s)^{B(n)}}.
\end{equation}
Taking expectation 
over
$X_j\sim\mathrm{Exp}(\lambda)$ and using its Laplace transform 
$\E[e^{-uX_j}] = \lambda/(\lambda+u)$ for $u \ge 0$, 
we have
\begin{equation}
\E\big[e^{-s\sum_{k=1}^{B(n)} Y_{(j-1)B(n)+k}}\big]
= 
\frac{(1+\sigma_0^2 s)^{-B(n)}}{1+\frac{s w_j P_n}{\lambda(1+\sigma_0^2 s)}}.
\end{equation}
Taking the product over all $M(n)$ blocks, and using $M(n)B(n)=n$,
\begin{equation}
\textstyle\E_1\big[e^{-s\sum_{i=1}^n Y_i}\big]
= (1+\sigma_0^2 s)^{-n}
\prod_{j=1}^{M(n)}\frac{1}{1+\tfrac{s w_j P_n}{\lambda(1+\sigma_0^2 s)}}.
\end{equation}

Therefore, with $t=\sigma_0^2 + d/\sqrt{n}$,
\begin{equation}
\P_{MD}
\le e^{sn\sigma_0^2 - n\log(1+\sigma_0^2 s) + ds\sqrt{n}
- \sum_{j=1}^{M(n)} \log\!\big(1+\tfrac{s w_j P_n}{\lambda(1+\sigma_0^2 s)}\big)}.
\end{equation}
Let $s={1}/{\sqrt n}$. Then $ds\sqrt n=d$, and since
$\sigma_0^2 s \to 0$, we may use the expansion $\log(1+u)=u-\tfrac{u^2}{2}+\mcO(u^3)$ to obtain
$sn\sigma_0^2-n\log(1+s\sigma_0^2)=\sigma_0^4/2 + o(1)$.
Moreover, since 
$\sum_j w_j=1$, it follows that $\prod_{j} (1+ w_i a) \geq 1+a$ for all $a\geq 0$ and hence 
\begin{equation}
\textstyle\sum_{j=1}^{M(n)} \log\!\Big(1+\tfrac{s w_j P_n}{\lambda(1+\sigma_0^2 s)}\Big) 
\ge\log\!\Big(1+\tfrac{s P_n}{\lambda(1+\sigma_0^2 s)}\Big).
\end{equation}
Together, these yield
\begin{equation}
\P_{MD}\le
{e^{{\sigma_0^4}/{2}+d+o(1)}}/{\Big(1+\tfrac{P_n}{\sqrt n\lambda(1+\sigma_0^2 /\sqrt n)}\Big)}.
\end{equation}
Hence, for any transmitted codeword, if $P_n/n=\omega(1/\sqrt{n})$ so that $P_n/\sqrt{n}\to\infty$, then $\P_{MD}\to0$. Combined with the ability to make $\P_{FA}$ any desired value, as described above, this means that any codeword sent with average power $P_n/n = \omega(1 / \sqrt{n})$ per symbol is detected by Willie with high probability. Hence, there does not exist a covert codebook with a positive fraction of codewords with power $\omega(1 / \sqrt{n})$ 
\ar{per symbol. Therefore, covertness requires the total transmitted energy of the codewords used with non-negligible probability to satisfy $P_n=\mcO(\sqrt n)$.} 

\ar{We next translate this power constraint into a bound on the number of decodable bits.
Let $W$ be uniformly distributed over the message set $\mathcal W_n$, and let $V^n=(V_1,\ldots,V_n)$ denote Bob's received sequence. Since providing Bob with the fading coefficients can only increase the mutual information, we may upper bound
$I(W;V^n)\le I(W;V^n\,|\{h_{\lceil i/B(n)\rceil}^{(a,b)}\}_{i=1}^n)$.
Conditioned on the fading coefficients, Bob's channel is Gaussian with received signal power scaled by $\big|h_{\lceil i/B(n)\rceil}^{(a,b)}\big|^2/d_{a,b}^{\alpha}$. Hence,
\[
	 I\!\big(W;V^n\big|\{h_{\lceil i/B(n)\rceil}^{(a,b)}\}_{i=1}^n\big)
	\le  \sum_{i=1}^{n}\!
	\E\Big[\!\log\!\Big(\!1+\tfrac{\big|h_{\lceil i/B(n)\rceil}^{(a,b)}\big|^2 |f_i|^2}{d_{a,b}^{\alpha}\sigma_b^2}\Big)\Big].
\]
Using $\log(1+x)\le x$ and $\E[|h_{\lceil i/B(n)\rceil}^{(a,b)}|^2]=1$,
\[
	\textstyle I(W;V^n)\le\frac{1}{d_{a,b}^{\alpha}\sigma_b^2}\sum_{i=1}^{n}|f_i|^2=\mcO(\sqrt n).
\]
By Fano's inequality, reliable decoding requires
\[
	\log |\mathcal W_n| \le I(W;V^n)+o(\log |\mathcal W_n|),
\]
and therefore $\log |\mathcal W_n|=\mcO(\sqrt n)$. Hence, Alice is restricted to transmitting at most $\mcO(\sqrt n)$ bits over $n$ channel uses.}
	
We emphasize that the converse argument does not rely on Willie employing an optimal detector; rather, it suffices to exhibit a specific test under which any transmission whose average power exceeds $\Theta(1/\sqrt{n})$ is reliably detectable, which in turn implies detectability under an optimal test. 
\ar{Note also that the proof does not require Willie to know Alice's transmitted codeword or its power allocation, and hence does not require him to know the induced distribution $\P_1$.
Since such knowledge could improve Willie's detection capability,} the converse results apply whether or not Willie knows $\P_1$.
\end{proof}

\section{Numerical Evaluation}
\label{sec:eval}

\ar{Our goals are twofold: to validate the square-root laws outlined in Section~\ref{sec:model} using Monte Carlo simulations and to demonstrate the suboptimality of the power detector.}
In what follows, we express Alice's variance as $\sigma_a^2 = c/n^{\rho}$, $\rho \in(0,1)$.  Figure~\ref{fig:phase2}(a) illustrates the error probability $\P_E=\P_{FA}+\P_{MD}$ of Willie's optimal detector, the Likelihood Ratio Test (LRT) as a function of $\rho$ for $\lambda=100$, $B(n)=1$, and four different values of $n$.  The parameters are set as follows: $c=0.03$, and a false alarm probability $\P_{FA}=0.01$.

As the value of $\rho$ increases, the error probability undergoes a noticeable transition, shifting from 0.01 to 1. This abrupt change is particularly evident at $\rho=0.5$, aligning with the condition $ \sigma^2_{a}(n)= c/\sqrt{n} $. This observation is consistent with the square root law. We also observe that, for the higher values of $n$, the phase transition is sharper.

\begin{figure}[h]
	\centering
	\begin{tabular}{cc}
	\includegraphics[width=0.21\textwidth]{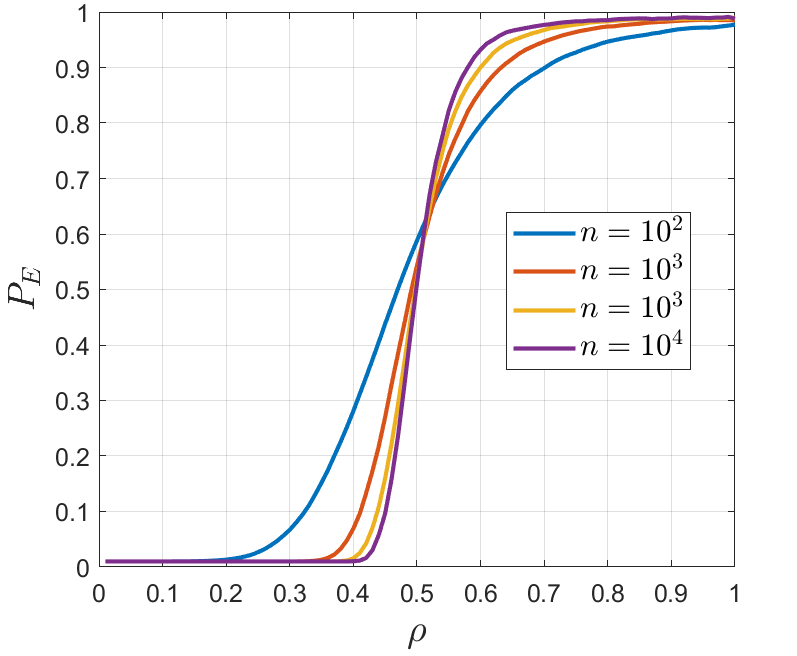} &
	\includegraphics[width=0.21\textwidth]{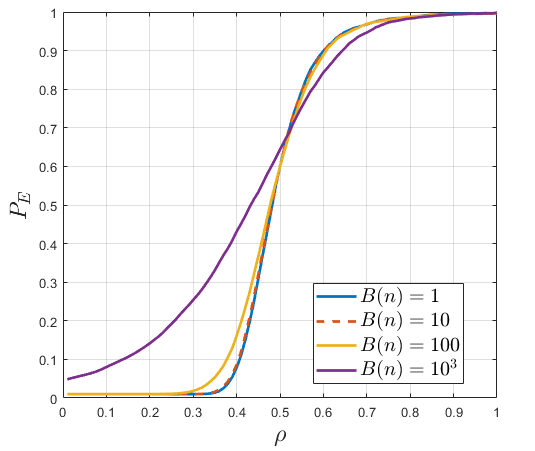}\\
		(a) & (b) \\
	\end{tabular}
	\caption{Phase transition analysis with (a) $B(n)=1$, $\lambda=100$, $\sigma^2_{w}=1$, $\P_{FA}=0.01$, and $\sigma^2_{a}=c/n^{-\rho}$ where $c=0.03$; (b) $n=1000$, $\lambda=100$, $\sigma^2_{w}=1$, $\P_{FA}=0.01$, and $\sigma^2_{a}=c/n^{-\rho}$ where $c=0.03$.
		As $\rho$ grows, $\sigma^2_{a}$ decreases and $\P_E$ transitions from 0.01 to 1. This transition occurs around $\rho=0.5$, in agreement with the square root law.}
	\label{fig:phase2}
\end{figure}

Figure~\ref{fig:phase2}(b) further explores the relationship between $\P_E$ and $\rho$ by holding $n$ constant at $1000$ and varying the number of blocks, while maintaining the same parameters. This variation aims to assess the impact of block size on $\P_E$ across three different block sizes. As the value of $\rho$ increases, 
$\P_E$ shifts from 0.01 to 1. Again, this change is evident at $\rho=0.5$. Also, a smaller block size results in a sharper phase transition. For fixed $n$, increasing $M(n)$ introduces more independent fading realizations, resulting in a sharper transition in $\P_E$ due to increased concentration of the likelihood ratio across blocks. This affects finite-$n$ behavior but not the $\Theta(\sqrt{n})$ scaling.

\ar{Next, we compare the LRT with the power detector (PD). Unlike the PD, which depends only on the total received energy, the LRT depends on the blockwise likelihoods and therefore on how the received energy is distributed across fading blocks. To illustrate this distinction, we consider a nonuniform power allocation in which Alice selects one of the $M(n)$ blocks uniformly at random and places all her codeword power in that block. Thus, the selected block has per-symbol power $M(n)\sigma_a^2$, while the remaining blocks contain no transmitted signal. Let $f_a(\mathbf z_i)$ denote the density of the selected block and let $f_0(\mathbf z_i)$ denote the noise-only density. Under $\mathrm H_1$, the joint density is
$\frac{1}{M(n)}\sum_{j=1}^{M(n)}f_a(\mathbf z_j)\prod_{i\ne j}f_0(\mathbf z_i)$,
and hence the likelihood ratio is $\Lambda(\mathbf z)=\frac{1}{M(n)}\sum_{j=1}^{M(n)}\frac{f_a(\mathbf z_j)}{f_0(\mathbf z_j)}$.}

Figure~\ref{fig:contourPlot} compares the LRT and PD for this setting. Figure~\ref{fig:contourPlot}(a) presents the LRT and PD decision boundaries for the two-block case, with both thresholds chosen to give $\P_{FA}=0.01$. The two decision boundaries are clearly distinct.

Since both tests are chosen to satisfy the same false-alarm probability, let
$A=\{\mathbf z:\Lambda(\mathbf z)>\eta\}$ denote the LRT rejection region and
let $B=\{\mathbf z:\sum_i \|\mathbf z_i\|^2>\tau\}$ denote the PD rejection
region, where $\eta$ and $\tau$ are chosen such that $\P_0(A)=\P_0(B)=0.01$.
Because the two boundaries in Figure~\ref{fig:contourPlot}(a) are distinct
and the distributions are continuous, the regions $A$ and $B$ differ on a set
of positive $\P_0$-probability. Therefore, using
$\Lambda=d\P_1/d\P_0$ and the fact that $\P_0(A)=\P_0(B)$, we have
\[
\textstyle	\P_1(A)-\P_1(B)
	=
	\int_{A\setminus B}(\Lambda-\eta)\,d\P_0
	+
	\int_{B\setminus A}(\eta-\Lambda)\,d\P_0
	>0 .
\]
Thus, the LRT achieves a strictly larger probability of detection than the
PD at the same false-alarm probability. Equivalently, the PD has a strictly
larger missed-detection probability and is therefore strictly suboptimal in
this example.

Figure~\ref{fig:contourPlot}(b) compares their detection performance for $B(n)=M(n)=40$, where Alice selects one of the 40 blocks uniformly and places all her power in that block. Under $\sigma_a^2=c/\sqrt{n}$, the LRT achieves a higher probability of true detection than the PD over the considered range of false-alarm probabilities.

\begin{figure}[h]
	\centering
	\begin{tabular}{cc}
		\includegraphics[width=0.20\textwidth]{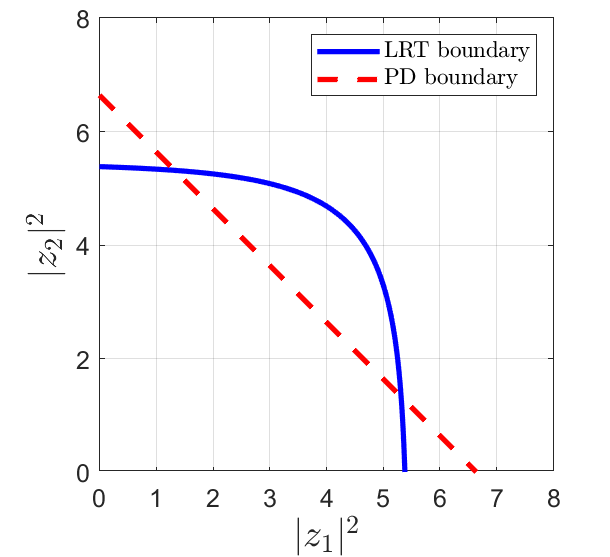} &
		\includegraphics[width=0.20\textwidth]{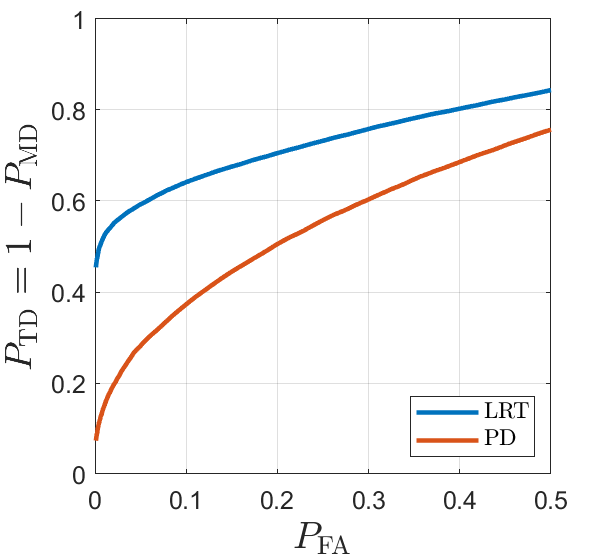}\\
		(a) & (b)
	\end{tabular}
	\caption{\ar{Comparison of the LRT and PD when Alice places all her codeword power in one block selected at random. (a) LRT and PD decision boundaries for $M(n)=2$, $\lambda=1$, $\sigma_w^2=1$, and average per-symbol power $\sigma_a^2=1$. Both thresholds satisfy $\P_{FA}=0.01$. (b) Probability of true detection $\P_{TD}$ versus $\P_{FA}$ for $B(n)=M(n)=40$, $n=1600$, and $\sigma_a^2=c/\sqrt{n}$ with $c=2.5$.}}
	\label{fig:contourPlot}
\end{figure}

\section{Conclusion}

We studied covert communication over block-fading channels, focusing on the asymptotic limits of covert bit transmission under optimal adversarial detection. We showed that the maximum number of bits Alice can reliably and covertly transmit over $n$ channel uses scales as $\Theta(\sqrt{n})$, thereby establishing the robustness of the square-root law in the presence of block fading with 
\ar{a growing} number of fading blocks. 
\ar{This result also indicates that, although block fading affects finite-length detector behavior and practical system performance, block fading alone does not fundamentally increase the covert throughput scaling when Willie employs an optimal detector. This observation is relevant to low-probability-of-detection wireless systems, including covert IoT, 5G/6G control signaling, and tactical communication networks.}
The challenging extension to the case when an uninformed jammer \cite{sobers2017covert} experiences block fading, as motivated by the discussion in \cite{goeckel2018covert}, is currently under consideration.  In addition, 
\ar{extensions to temporally correlated fading and to other fading distributions, such as Rician or Nakagami fading, provide avenues for future work.}

\bibliographystyle{IEEEtran}

\bibliography{IEEEabrv,main}

\end{document}